\documentclass[letterpaper,10pt, conference]{_aux/ieeeconf}
\usepackage{times}
\usepackage{setspace}
\usepackage{url}
\spacing{1}
\usepackage[utf8]{inputenc}
\usepackage[T1]{fontenc}
\usepackage{graphicx}		
\usepackage{wrapfig}
\usepackage{sidecap} 
\usepackage[export]{adjustbox}
\usepackage{subcaption}
\usepackage[format=plain,font=small,labelfont=bf,labelsep=period]{caption}
\usepackage{float}

\usepackage{hyperref}

\usepackage{amsmath} 
\usepackage{amssymb}  
\usepackage{amsthm}
\usepackage{mathtools}
\usepackage[normalem]{ulem}
\usepackage{paralist}	
\usepackage[space]{grffile} 
\usepackage{color}
\let\labelindent\relax
\usepackage{enumitem}
\usepackage{bm}
\usepackage{cancel}
\usepackage{hhline}
\usepackage[c2 , nocomma]{optidef}
\usepackage{oubraces}
\usepackage{algorithmic}
\usepackage{graphicx}
\usepackage{textcomp}

\usepackage{amsfonts}
\usepackage{lipsum}
\usepackage{cleveref} 
\usepackage{cite}

\usepackage{comment}

\newtheorem{theorem}{Theorem}
\newtheorem{corollary}{Corollary}
\newtheorem{lemma}{Lemma}
\theoremstyle{definition}
\newtheorem{definition}{Definition}
\theoremstyle{remark}
\theoremstyle{definition}
\theoremstyle{definition}
\newtheorem{proposition}{Proposition}
\newtheorem{example}{Example}

\newcommand{\R}{\mathbb{R}}

\newcommand{\x}{\boldsymbol{x}}

\newcommand{\sspace}{\mkern9mu}
\newcommand{\nspace}[1]{\mkern#1mu}

\newcommand{\norm}[1]{\left\Vert #1 \right\Vert}

\newcommand{\Cs}{\mathcal{C}_{\rm S}} 
\newcommand{\Cb}{\mathcal{C}_{\rm B}} 
\newcommand{\Cbi}{\mathcal{C}_{\rm I}} 

\newcommand{\Tt}{T} 

\newcommand{\ub}{\boldsymbol{k}_{\rm b}} 
\newcommand{\um}{\boldsymbol{k}_{\rm m}} 
\newcommand{\udes}{\boldsymbol{k}_{\rm p}} 

\newcommand{\phinb}[2]{\boldsymbol{\phi}_{\rm b} (#1, #2)}

\newcommand{\phinom}{\phinb{\tau}{\boldsymbol{x}}}
\newcommand{\phinomT}{\phinb{\Tt}{\boldsymbol{x}}}

\newcommand{\stmnom}{\boldsymbol{\Phi}_{\rm b}(\tau,\boldsymbol{x})}
\newcommand{\stmnomT}{\boldsymbol{\Phi}_{\rm b}(\Tt,\boldsymbol{x})}
\newcommand{\jac}{F_{\rm cl}}

\definecolor{darkblue}{RGB}{0,0,102}
\definecolor{lightblue}{RGB}{77,77,148}

\definecolor{gold}{RGB}{234, 170, 0}
\definecolor{metallic_gold}{RGB}{139, 111, 78}

\DeclareMathOperator{\diag}{diag}

\usepackage{xfrac}

\usepackage{dblfloatfix}

\usepackage{pdfpages}
\usepackage{graphicx}

\usepackage{enumitem}

\newcommand{\bx}{\mathbf{x}}

\def\BibTeX{{\rm B\kern-.05em{\sc i\kern-.025em b}\kern-.08em
T\kern-.1667em\lower.7ex\hbox{E}\kern-.125emX}}

\IEEEoverridecommandlockouts

\begin{document}

\title{\LARGE \vspace{-.65cm}
\textbf{Safety-Critical Control with Bounded Inputs: \\ A Closed-Form Solution for Backup Control Barrier Functions}
}

\author{David E. J. van Wijk$^{1}$,
Ersin Da\c{s}$^{1}$, Tamas G. Molnar$^{2}$, Aaron D. Ames$^{1}$, and Joel W. Burdick$^{1}$
\thanks{*This work was supported by DARPA under the LINC Program and the AFOSR Test and Evaluation program, grant FA9550-22-1-0333.}
\thanks{
$^{1}$Mechanical and Civil Engineering, California Institute of Technology, Pasadena, CA 91125, USA, \texttt{\{vanwijk, ersindas, ames, jburdick\}@caltech.edu}.}
\thanks{$^{2}$Mechanical Engineering, Wichita State University, Wichita, KS 67260, USA, \texttt{tamas.molnar@wichita.edu}.}
}

\maketitle

\begin{abstract}
Verifying the safety of controllers is critical for many applications, but is especially challenging for systems with bounded inputs. Backup control barrier functions (bCBFs) offer a structured approach to synthesizing safe controllers that are guaranteed to satisfy input bounds by leveraging the knowledge of a backup controller. While powerful, bCBFs require solving a high-dimensional quadratic program at run-time, which may be too costly for computationally-constrained systems such as aerospace vehicles. We propose an approach that optimally interpolates between a nominal controller and the backup controller, and we derive the solution to this optimization problem in closed form. We prove that this closed-form controller is guaranteed to be safe while obeying input bounds. We demonstrate the effectiveness of the approach on a double integrator and a nonlinear fixed-wing aircraft example.
\end{abstract}
\section{Introduction}
Controlling dynamic systems
with constraints that arise from
physical limitations, environmental interactions, or user-defined requirements is a fundamental aspect of control applications.
As such, constrained control \cite{mayne2000}---also called {\em safety-critical control} when constraints encode safety requirements \cite{hsu2023safety}---has become a top priority in many autonomous or semi-autonomous cyber-physical systems ranging from robotics to aeronautics.
{\em Control barrier functions (CBFs)} \cite{ames_2017} have emerged as a framework for synthesizing safe controllers
that guarantee the forward invariance of a specified safe set.
While CBFs have demonstrated success in ensuring safety in many domains \cite{wabersich2023data,garg2024advances}, they also highlight certain challenges.
Namely, conventional CBFs require one to verify that the safe control inputs are feasible at all points within the safe set, which may be difficult in the presence of input constraints.

\begin{figure}[t]
    \centering 
    \vskip  2.2mm
    \includegraphics[width=1\linewidth]{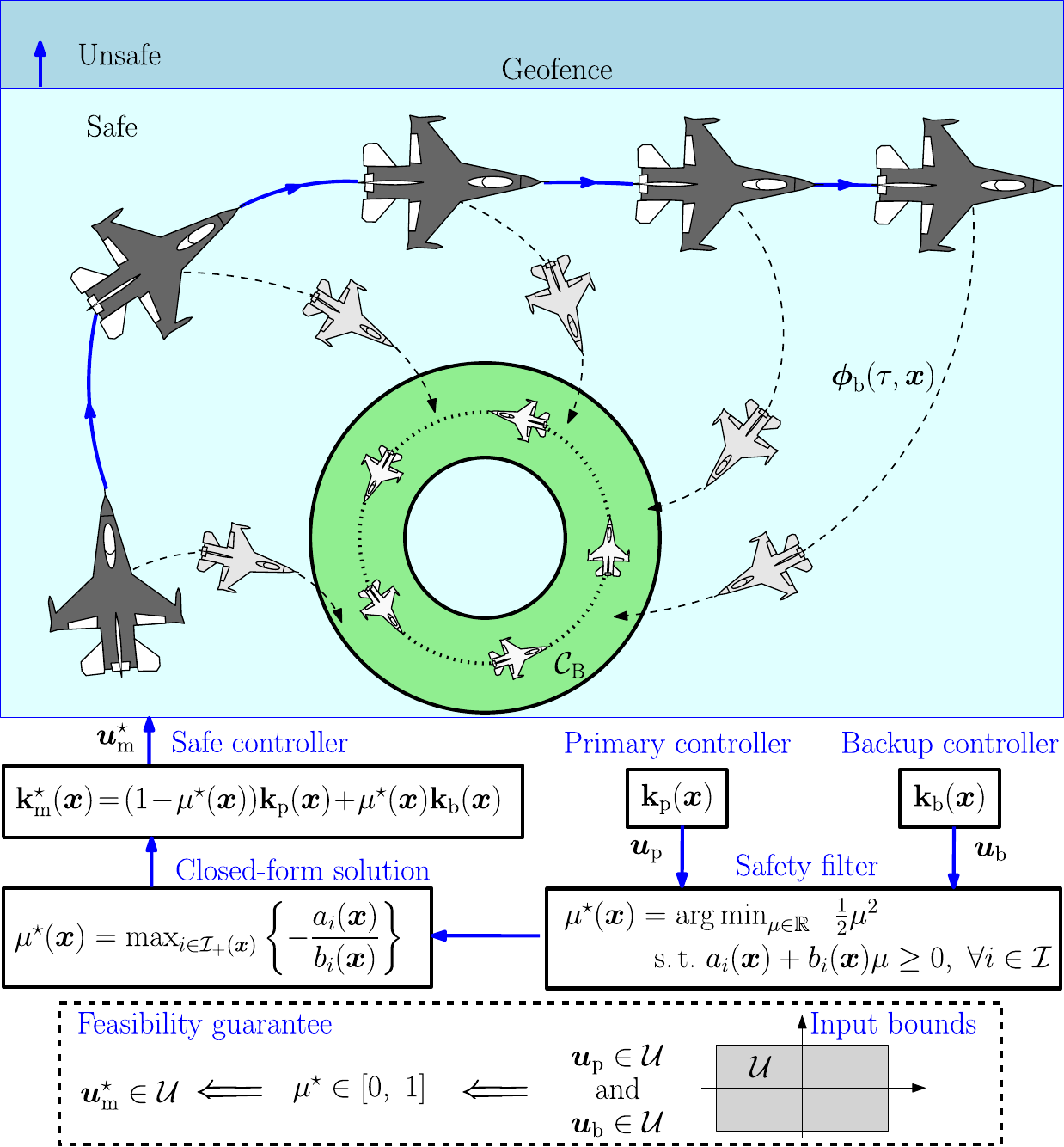}
    \caption{Overview of the proposed safety-critical control synthesis framework, illustrated on a fixed-wing aircraft geofencing scenario. Our \textit{optimally interpolated (OI)} controller guarantees the safety of nonlinear affine systems with input bounds and is obtained in closed-form.}
    \label{fig: env}
\end{figure}

Motivated by this feasibility challenge, {\em backup control barrier functions (bCBFs)} \cite{gurriet_online_2018} have been proposed to provide theoretical guarantees for the safety of control systems with input bounds. The bCBF technique is applicable for complex nonlinear systems, and has shown success across a wide range of applications \cite{dunlap2022comparing,vanWijk_DRbCBF_24,hobbs2023rta,gurriet_thesis,ko2024backup,zeng2025onboard,janwani2024learning,rabiee2025soft,rivera2024forward}.
However, a potential downside of this approach is its online computational requirements: it involves solving a quadratic program (QP) with a large number of constraints at run-time, and to construct these constraints, one must
forward integrate a number of ordinary differential equations (ODEs) whose number increases quadratically with the state dimension. This may not be feasible for high-dimensional nonlinear systems with limited online computational resources, such as aerospace vehicles. Thus, many safety-critical control solutions for aerospace vehicles have utilized CBFs rather than bCBFs (see e.g., \cite{altunkaya2025nuisancefree,squires2022composition,zheng2023constrained,breeden2020,molnar2025fixedwing}).  

To remedy this, the authors of \cite{singletary2022_on,singletary_ICRA2020} have introduced a technique for safe controller synthesis that builds on the concept of bCBFs without requiring optimization for the safe input.
Instead, this method -- which we call {\em safe blending} -- utilizes a smooth blending function to safely
combine a nominal controller and a backup controller.
This approach does not require expensive gradient computations, nor does it involve solving a QP online.
While effective, this solution is inherently sub-optimal, as the blending function is chosen subjectively and must be hand-tuned. Furthermore, this approach may exhibit undesirable oscillations
when the nominal and backup controllers act against each other,
resulting in controller chattering and degraded system performance. 

Inspired by these previous works, we propose a middle ground between the bCBF and the blended approach. Our main contribution is the \textit{optimally interpolated (OI)} controller---a novel safety-critical control technique for systems with bounded inputs (see \Cref{fig: env}). The proposed method optimally interpolates between a nominal controller and a pre-verified backup controller to ensure that the dynamical system remains within a controlled invariant safe set for all time.
We derive a closed-form expression for the optimal safe controller, and prove that this solution is feasible even in the presence of input bounds.
We demonstrate that the proposed OI controller does not experience the undesirable oscillations observed for the blending approach.
Furthermore, compared to the bCBF approach, the number of ODEs to be integrated is smaller, while the closed-form expression eliminates the need for solving a high-dimensional QP numerically. 

The rest of the paper is organized as follows. Section~\ref{sec:pre} overviews CBFs and bCBFs. Section~\ref{sec:OG_blending} discusses the function-based controller blending method.
In Section~\ref{sec:theory}, we present the proposed
optimally interpolated safety-critical control approach.
In Section~\ref{sec:sim}, we cover the results of our numerical simulations, and we conclude with Section~\ref{sec:conc}.
\section{Preliminaries}
\label{sec:pre}
\subsection{Control Barrier Functions}
Consider a nonlinear control affine system of the form
\begin{align} \label{eq:affine-dynamics}
    \dot{\x} = f(\x) + g(\x)\boldsymbol{u}, \nspace{8}
    \x \in \mathcal{X} \subseteq \R^n, \nspace{8}
    \boldsymbol{u} \in \mathcal{U} \subseteq \R^m,
\end{align}
where ${f:\mathcal{X} \to \R^n}$ and ${g:\mathcal{X} \to \R^{n \times m}}$ are smooth functions. We assume that $\mathcal{U}$ is an $m$-dimensional convex polytope.
For an initial condition ${\x(0) = \x_0 \in \mathcal{X}}$, if $\boldsymbol{u}$ is given by a locally Lipschitz feedback controller ${\boldsymbol{k}:\mathcal{X} \to \mathcal{U}}$, ${\boldsymbol{u}=\boldsymbol{k}(\x)}$, the closed-loop system has a unique solution.

Safety is defined by membership to a set $\Cs$, and safe controllers render this safe set forward invariant.
\begin{definition}
A set ${\mathcal{C} \subseteq \R^n}$ is \textit{forward invariant} along the closed-loop system if ${\x(0) \in \mathcal{C} \implies \x(t) \in \mathcal{C},}$ for all ${t > 0}$.
\end{definition}
Consider the safe set ${\Cs \triangleq \{\x \in \mathcal{X} : h(\x) \ge 0\}}$ as the 0-superlevel set of a continuously differentiable function ${h : \mathcal{X} \to \R}$,
where the gradient of $h$ along the boundary of $\Cs$ is nonzero.
The function $h$ is a CBF \cite{ames_2017} for \eqref{eq:affine-dynamics} on $\Cs$ if there exists a class-$\mathcal{K}_{\infty}$ function\footnote{A function ${\alpha : \R_{\ge 0} \to \R_{\ge 0}}$ is of class-$\mathcal{K}_{\infty}$ if it is continuous, strictly increasing, ${\alpha(0)=0}$, and $\text{lim}_{r \to \infty} \nspace{2}\alpha(r) = \infty$.} $\alpha$ such that for all $\x \in \Cs$
\begin{equation*}
    \sup_{\boldsymbol{u} \in \mathcal{U}} \big [ \dot{h}(\x,\boldsymbol{u}) \triangleq
    \nabla h(\x) \cdot \big( f(\x) + g(\x) \boldsymbol{u} \big) \big ]
    >  -\alpha(h(\x)).
\end{equation*}
\begin{theorem}[\hspace{-0.01em}\cite{ames_2017}] \label{thm: cbf}
If $h$ is a CBF for \eqref{eq:affine-dynamics} on $\Cs$, then any locally Lipschitz controller $\boldsymbol{k}:\mathcal{X} \to \mathcal{U}$, $\boldsymbol{u}=\boldsymbol{k}(\x)$ satisfying 
\begin{align} \label{eq: cbf_condition}
    \dot{h}(\x, \boldsymbol{u}) \ge -\alpha(h(\x)),
\end{align}
for all $\x \in \Cs$ renders the set $\Cs$ forward invariant.
\end{theorem}

For an arbitrary \textit{primary} or \textit{nominal controller}, ${\boldsymbol{k}_{\rm p} : \mathcal{X} \to \mathcal{U}}$, one can ensure the safety of \eqref{eq:affine-dynamics} by solving a quadratic program (QP) for the safe control, ${\boldsymbol{k}_{\rm safe}:\mathcal{X}\to\mathcal{U}}$:
\begin{align} 
    \boldsymbol{k}_{\rm safe}(\x) = \underset{\boldsymbol{u} \in \mathcal{U}}{\text{argmin}} \mkern9mu &
    \left\Vert \boldsymbol{u} - \boldsymbol{k}_{\rm p}(\x) \right\Vert^{2}
    \tag{CBF-QP} \label{eq:cbf-qp} \\
    \text{s.t.} \quad &
    \dot{h}(\x, \boldsymbol{u}) \ge -\alpha(h(\x)).
    \nonumber
\end{align}
If $h$ is a CBF, then the~\eqref{eq:cbf-qp} is feasible, and when ${\mathcal{U} = \R^m}$, its closed-form solution becomes~\cite{cohen2023smooth}
\begin{align} \label{eq:safetyfilter}
    \boldsymbol{k}_{\rm safe}(\x) & = \udes(\x) + \hat{\mu}(\x) \big( \nabla h(\x) \cdot g(\x) \big)^\top, \\
    \hat{\mu}(\x) & = \lambda \big( \hat{a}(\x), \hat{b}(\x) \big), \nonumber
\end{align}
where $\hat{a}$ and $\hat{b}$ stem from the coefficients in the constraint:
\begin{align*}
    \hat{a}(\x) & = \dot{h} \big( \x, \udes(\x) \big) + \alpha \big( h(\x) \big), \\
    \hat{b}(\x) & = \| \nabla h(\x) \cdot g(\x) \|^2,
\end{align*}
whereas $\lambda$
is defined using ${{\rm ReLU}(\cdot) \triangleq \max\{0,\cdot\}}$ as
\begin{equation} \label{eq:lambda_qp}
    \lambda(a,b) = 
    \begin{cases}
        0 & {\rm if}\ b \leq 0, \\
        {\rm ReLU}(-a/b) & {\rm if}\ b > 0.
    \end{cases}
\end{equation}

Ensuring the feasibility of the \eqref{eq:cbf-qp} can be challenging, especially for high dimensional systems. This motivates the use of an extension of CBFs known as backup CBFs.

\subsection{Backup Control Barrier Functions} \label{sec:bCBF}

Backup CBFs (bCBFs) \cite{gurriet_online_2018,gurriet_scalable_2020} can overcome the feasibility issue of traditional CBFs by constructing controlled invariant sets at run-time for systems with input constraints. 
\begin{definition}
A set $\mathcal{C} \subseteq \R^n$ is \textit{controlled invariant} if there exists a controller ${\boldsymbol{k}:\mathcal{X} \to \mathcal{U}}$, ${\boldsymbol{u}=\boldsymbol{k}(\x)}$ rendering $\mathcal{C}$ forward invariant for \eqref{eq:affine-dynamics}, with ${\boldsymbol{u} \in \mathcal{U}}$. 
\end{definition}

First, one must obtain a controlled invariant subset of $\Cs$, known as the \textit{backup set}:
\begin{align}
    \Cb \triangleq \{\x \in \mathcal{X} : h_{\rm b}(\x) \ge 0\} \subseteq \Cs,
\end{align}
where ${h_{\rm b} : \mathcal{X} \to \R}$ is continuously differentiable with a nonzero gradient along the boundary of $\Cb$.
Second, one must define a smooth \textit{backup control law} $\ub : \mathcal{X} \to \mathcal{U}$ which renders the backup set forward invariant along the
system
\begin{align}\label{eq: f_cl}
    \dot{\x} = f_{\rm cl}(\x) \triangleq f(\x) + g(\x)\ub(\x).
\end{align}

Finding a small controlled invariant subset of $\Cs$ is typically much easier than verifying that $\Cs$ itself is controlled invariant. For example, a backup set can often be defined by a level set of a quadratic Lyapunov function centered on a stabilizable equilibrium point for the linearized dynamics, and can be rendered forward invariant by a simple feedback controller \cite{gurriet_scalable_2020}.
One could directly constrain the system \eqref{eq:affine-dynamics} to operate within $\Cb$ to guarantee safety. However, as $\Cb$ may be a very small subset of $\Cs$, this would lead to conservative behavior.
Instead, it is possible to expand this set implicitly, thereby obtaining a new, larger controlled invariant set.

To construct a controlled invariant set online, we allow the system to evolve beyond $\Cb$ by forward integrating the backup dynamics~\eqref{eq: f_cl} over a finite horizon. If the system can stay in $\Cs$ and safely reach $\Cb$ from the current state using $\ub$, this state is classified as safe. More precisely, the controlled invariant set, ${\Cbi \!\subseteq\! \Cs}$, is defined as
\begin{align} \label{def:C_BI}
    \Cbi \triangleq \left\{ \x \in \mathcal{X} \,\middle|\, 
    \begin{array}{c}
    h(\phinom) \geq 0, \forall \nspace{1} \tau \in [0,T], \\
    h_{\rm b}(\phinomT) \geq 0 \\
    \end{array}
    \right\},
\end{align}
where $\phinom$ is the \textit{flow} of the backup system~\eqref{eq: f_cl} over the interval $\tau\in[0,T]$ for a horizon $T > 0$ starting at state $\x$:
\begin{align} \label{eq: nomFlow}
    \frac{\partial}{\partial \tau}{\boldsymbol{\phi}_{\rm b}}(\tau,\x) = f_{\rm cl}(\phinom), \sspace \phinb{0}{\x} = \x.
\end{align}
This definition gives recursive feasibility guarantees under input constraints.
\begin{lemma}[\!\!\cite{gurriet_scalable_2020}\!{\cite[Lem.~1]{tamasACC_ROM_bCBF}}]  \label{lemma: CBI_controlledInv}
The set $\Cbi$ is controlled invariant, and the backup controller $\ub$ renders $\Cbi$ forward invariant along \eqref{eq: f_cl} as
${\x \in \Cbi \implies \phinom \in \Cbi \subseteq \Cs, \forall \tau \geq 0}$.
\end{lemma}
As a direct consequence of \Cref{lemma: CBI_controlledInv}, the backup control law can be shown to satisfy safety constraints similar to \eqref{eq: cbf_condition}.
\begin{lemma}[\!\!\cite{gurriet_scalable_2020}\!{\cite[Lem.~2]{tamasACC_ROM_bCBF}}] \label{lemma: backup_constraint}
    There exist class-$\mathcal{K}_{\infty}$ functions $\alpha$, $\alpha_{\rm b}$ such that
    \begin{subequations} \label{eq:backup_constraint}
    \begin{align} 
        \dot{h} \big( \phinom,\ub(\x) \big)
        &\ge - \alpha \big(h(\phinom)\big),  \\ 
        \!\! \dot{h}_{\rm b} \big( \phinomT, \ub(\x) \big) 
        &\ge - \alpha_{\rm b} \big(h_{\rm b}(\phinomT)\big), 
    \end{align}
    \end{subequations}
    for all ${\tau \in [0,T]}$ and ${\x\in\mathcal{\Cbi}}$.
\end{lemma}
This therefore leads to the main result of backup CBFs.
\begin{theorem}[\!\!\cite{gurriet_scalable_2020}] \label{thm: backup_cbf}
    There exist class-$\mathcal{K}_{\infty}$ functions $\alpha$, $\alpha_{\rm b}$ and a locally Lipschitz controller ${\boldsymbol{k} \!:\! \mathcal{X} \!\to\! \mathcal{U}}$, ${\boldsymbol{u}\!=\!\boldsymbol{k}(\x)}$ satisfying\!\!
    \begin{subequations} \label{eq: backupCBFMain}
    \begin{align} 
        \dot{h}(\phinom,\boldsymbol{u})
        &\ge - \alpha \big(h(\phinom)\big),\ \forall \tau \in [0,T], \label{eq: htraj_nom}  \\ 
        \!\! \dot{h}_{\rm b}(\phinomT, \boldsymbol{u}) 
        &\ge - \alpha_{\rm b} \big(h_{\rm b}(\phinomT)\big), \label{eq: hb_nom}
    \end{align}
    \end{subequations}
    for all
    ${\x\in\mathcal{\Cbi}}$. Such a controller renders ${\Cbi \subseteq \Cs}$ forward invariant and satisfies input constraints.
\end{theorem}
\begin{proof}
    By \Cref{lemma: CBI_controlledInv}, the backup control law $\ub$ is one such Lipschitz controller. Applying \Cref{thm: cbf} to system \eqref{eq:affine-dynamics}, the satisfaction of \eqref{eq: backupCBFMain} yields the forward invariance of $\Cbi$.
\end{proof}
Because \eqref{eq: htraj_nom} represents an infinite number of constraints, in practice the constraint is discretized and enforced at discrete points along the flow, such that ${\tau \in \{0, \Delta, \dots, \Tt \}}$ where ${\Delta > 0}$ is a discretization step satisfying ${T/\Delta \in \mathbb{N}}$.
Then, one can construct a point-wise optimization problem for the least-invasive safe control signal, like the \eqref{eq:cbf-qp}:
\begin{align*} 
    \boldsymbol{k}_{\rm safe}(\x) = \underset{\boldsymbol{u} \in \mathcal{U}}{\text{argmin}} \mkern9mu &
    \left\Vert \boldsymbol{u}- \boldsymbol{k}_{\rm p}(\x)\right\Vert^{2} \quad \tag{bCBF-QP} \label{eq:bcbf-qp} \\
    \text{s.t.  } 
    & \eqref{eq: htraj_nom}, \ \eqref{eq: hb_nom}.
\end{align*}
\begin{corollary}
    There exist class-$\mathcal{K}_{\infty}$ functions $\alpha$, $\alpha_{\rm b}$ such that the optimization problem posed in \eqref{eq:bcbf-qp} is feasible for all $\x \in \Cbi$,
    even if $h$ is not a CBF.
\end{corollary}
Thus, bCBFs are powerful tools for guaranteeing safety of nonlinear affine systems in the presence of input constraints. 
However, the price for the feasibility guarantees endowed by bCBFs is additional computation. 
Computing the total derivatives of $h$ and $h_{\rm b}$:
\begin{subequations}
\begin{align} 
    \dot{h}(\phinom,\boldsymbol{u}) &=  \nabla h(\phinom) \cdot \stmnom\dot{\x}, \\ 
    \dot{h}_{\rm b}(\phinomT,\boldsymbol{u}) &= \nabla h_{\rm b}(\phinomT) \cdot \stmnomT \dot{\x},
\end{align}
\end{subequations}
where ${\dot{\x}= f(\x) + g(\x)\boldsymbol{u}}$, and ${\stmnom \triangleq \frac{\partial \phinom}{\partial \x}}$,
requires solving a coupled initial value problem for the sensitivities of the flow to changes in the initial condition.
Namely, the sensitivity matrix $\stmnom$ is the solution to
\begin{equation} \label{eq: stm_nominal}
  \begin{gathered} 
    \frac{\partial}{\partial \tau}{{\boldsymbol{\Phi}}}_{\rm b}(\tau, \x) \!=\! \jac(\phinom)\stmnom, \nspace{6}
    \boldsymbol{\Phi}_{\rm b}(0,\x) \! =  \!\boldsymbol{I},
\end{gathered}
\end{equation}
where $\jac(\x)=\partial f_{\rm cl}(\x)/\partial \x$ is the Jacobian of $f_{\rm cl}$ in~\eqref{eq: f_cl}, that is evaluated at $\phinom$, and $\boldsymbol{I}$ is the ${n \!\times\! n}$ identity matrix. 

To compute the constraints in \eqref{eq: backupCBFMain}, at each time step one must forward integrate ${n + n^2}$ ordinary differential equations (ODEs) in~\eqref{eq: nomFlow} and~\eqref{eq: stm_nominal}. For many applications this may be allowable, but for resource-constrained platforms, such as aerospace systems, computing these gradients may be a bottleneck,
especially if the backup horizon is long. Further, formal system certification requirements may exclude methods that
numerically solve a high-dimensional QP online.
\section{Function-Based Controller Blending}\label{sec:OG_blending}

While bCBFs offer an elegant solution to feasible safety-critical control, their main downside is
the additional computation requirements.
To overcome this downside, the authors of \cite{singletary2022_on} offer a solution deemed \textit{safe blending}, which seeks to mix the control inputs from the nominal and the backup controller, to obtain a safe but more performant controller than just the backup controller alone. This blended or \textit{mixed} controller $\boldsymbol{k}_{\rm m} :\mathcal{X}\to\mathcal{U}$ is written as
\begin{equation}
\label{eq:safe_ble1}
    \boldsymbol{k}_{\rm m}(\boldsymbol{x}) = (1 - \mu(\boldsymbol{x}) )\boldsymbol{k}_{\rm p}(\boldsymbol{x}) + \mu(\boldsymbol{x}) \ub(\boldsymbol{x}),
\end{equation}
where $\mu: \R^n \to [0, 1]$ is a blending function that regulates the portion that the backup and primary control signals contribute to the mixed control signal. 
It is trivial to verify that restricting the range of $\mu$ from ${[0, 1]}$ guarantees that for all $\x \in \mathcal{X}$, ${\boldsymbol{k}_{\rm m}(\boldsymbol{x}) \in \mathcal{U}}$ as long as ${\boldsymbol{k}_{\rm p}(\boldsymbol{x}), \ub(\boldsymbol{x}) \in \mathcal{U}}$ and $\mathcal{U}$ is convex.

To ensure safety with the mixed controller, the authors of \cite{singletary2022_on} define a function $h_{\rm I}$ associated with the set $\Cbi$ in  \eqref{def:C_BI} as
\begin{equation*}
\label{eq:h_I}
h_{\rm I}(\boldsymbol{x}) \triangleq \min \left \{  \min_{\tau\in[0,T]} h(\phinom), h_{\rm b}(\phinomT) \right \}.
\end{equation*}
Using this definition, a blending function is proposed as
\begin{equation}
\label{eq:blend_1}
\mu(\boldsymbol{x}) = \Lambda(h_{\rm I}({\boldsymbol{x}})),
\end{equation}
where ${\Lambda : \mathbb{R} \to [0,1]}$ is a continuous function that satisfies ${\Lambda(r) = 1}$ when ${r \leq 0}$, and ${\frac{d \Lambda}{d r} < 0}$ when ${r > 0}$.  
For example,
\begin{equation}
\label{eq:tun_par}
    \Lambda(h_{\rm I}({\boldsymbol{x}})) = {\rm e}^{-\eta \max \{h_{\rm I}(\boldsymbol{x}),\!~0 \} },
\end{equation}
with ${\eta >0}$. This particular choice of blending function and blending controller is provably safe.
\begin{proposition}[\hspace{-0.01em}{\cite[Proposition 1]{singletary_ICRA2020}}]
\textit{
    Any locally Lipschitz mixed controller ${\um:\mathcal{X} \to \mathcal{U}}$ given by \eqref{eq:safe_ble1}, with a blending function $\mu$ in \eqref{eq:blend_1}
    and ${\Lambda : \mathbb{R} \to [0,1]}$ satisfying ${\Lambda(0)=1}$, renders $\Cbi$ forward invariant along \eqref{eq:affine-dynamics}.
}
\begin{proof}
    At the boundary of $\Cbi$, where ${h_{\rm I}(\boldsymbol{x}) = 0}$, the blending parameter in~\eqref{eq:safe_ble1} is ${\mu(\boldsymbol{x}) = 1}$ and the mixed controller in \eqref{eq:blend_1} reduces to ${\boldsymbol{k}_{\rm m}(\boldsymbol{x}) = \ub(\boldsymbol{x})}$.
    Therefore, because the set $\Cbi$ is forward invariant under the backup controller based on Lemma~\ref{lemma: CBI_controlledInv},
    $\Cbi$ is rendered forward invariant under $\um$ as well. Furthermore, $\um(\boldsymbol{x}) \in \mathcal{U}$ for all ${\boldsymbol{x} \in \mathcal{X}}$ by design.
\end{proof}
\end{proposition}

The mixed controller in \eqref{eq:safe_ble1} is not only provably safe, but it also circumvents the need to solve an optimization problem online and to calculate the sensitivity of the backup flow by solving $n^2$ ODEs in~\eqref{eq: stm_nominal}.
However, the mixed controller can lead to undesirable behavior near the boundary of $\Cbi$ if the goals of the primary controller and backup controller conflict.
We use a simple example to highlight this point.

\begin{example}[Double Integrator]\label{ex:db_int_osc}
Consider a double integrator system as in \cite{vanWijk_DRbCBF_24} given by
\begin{align} \label{eq: db_int}
    \dot{\boldsymbol{x}} = 
    \begin{bmatrix}
        x_2 &
        u
    \end{bmatrix}^{\top},
\end{align}
with position $x_1$, velocity $x_2$, state ${\boldsymbol{x} = [x_1~x_2]^\top \in \mathbb{R}^2}$, and control input ${u \in \mathcal{U} = [-1, 1]}$. The safe set is defined as the left half-plane
\begin{align*}
    \Cs \triangleq \{\boldsymbol{x} \in \mathbb{R}^2 : -x_1 \geq 0 \}.
\end{align*}
The backup control law, ${\boldsymbol{k}_{\rm b}(\boldsymbol{x}) = -1}$, brings the system to the backup set ${\Cb \triangleq \{\boldsymbol{x} \in \mathbb{R}^2 : -x_1 \geq 0, -x_2 \geq 0 \}}$ and renders $\Cb$ forward invariant.
The primary controller, ${\boldsymbol{k}_{\rm p}(\boldsymbol{x}) = 1}$, drives \eqref{eq: db_int} to the unsafe right half-plane.

\begin{figure}
\centering
\begin{subfigure}{1\columnwidth}
\hspace{-.364cm}
\centerline{\includegraphics[width=1\columnwidth]{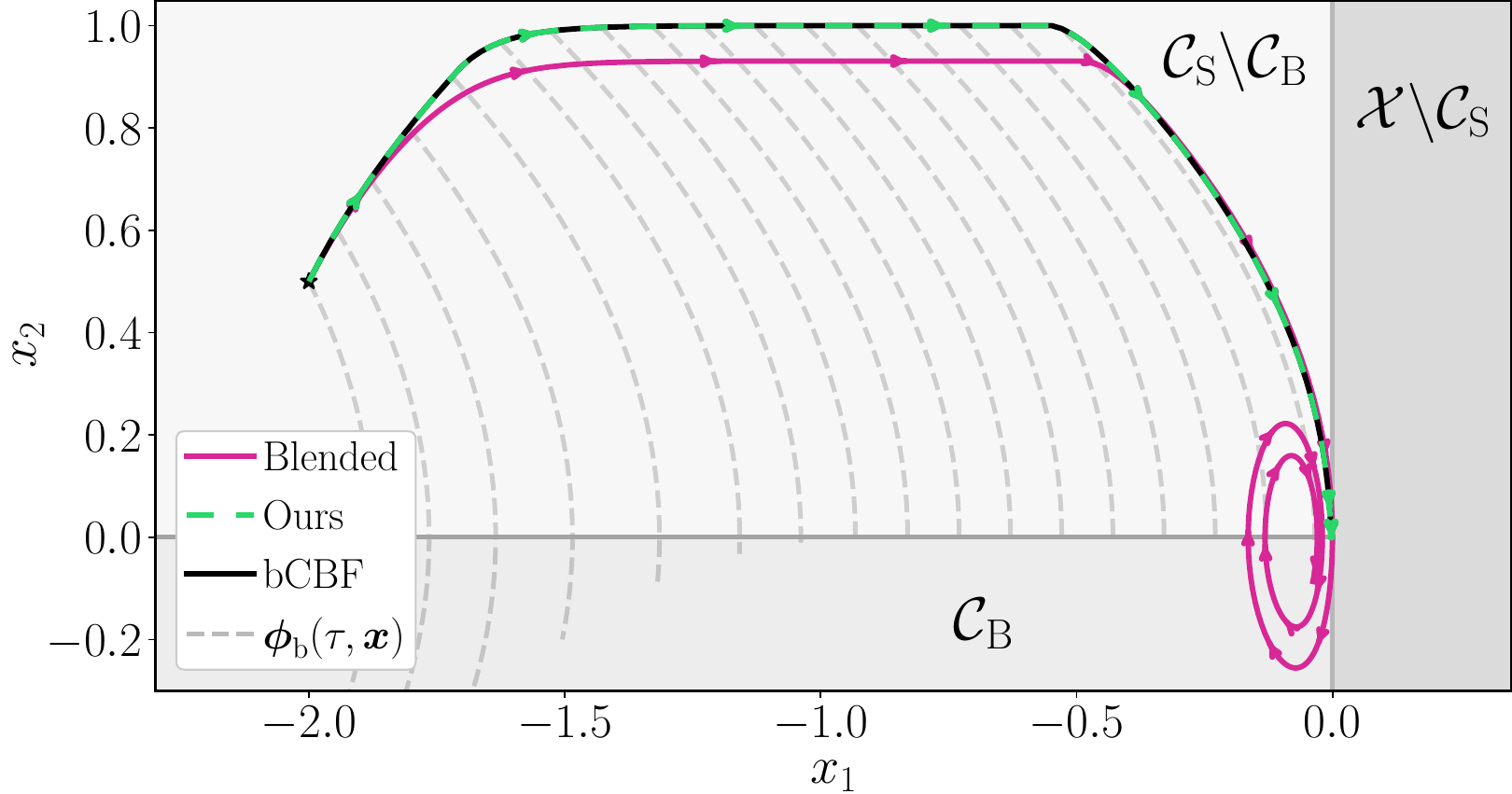}}
\vspace{.2cm}
\end{subfigure}
\begin{subfigure}{1\columnwidth}
\centerline{\includegraphics[width=1\columnwidth]{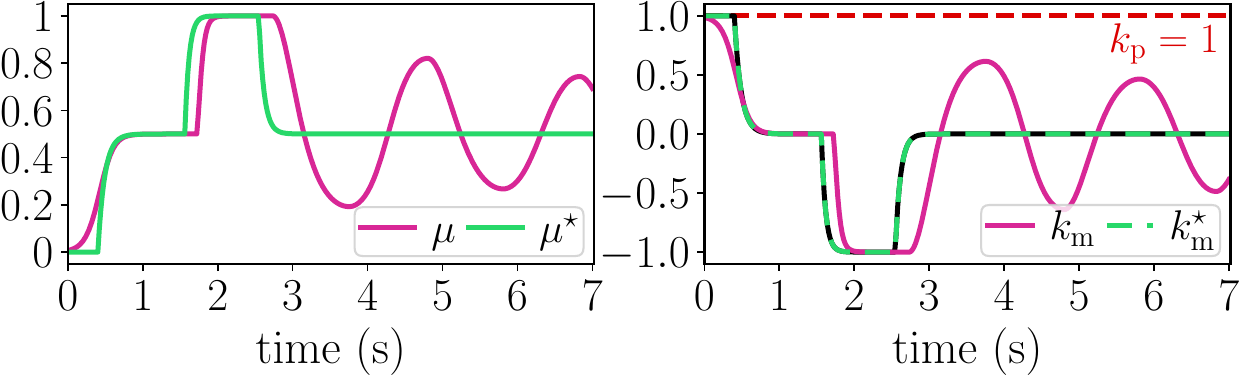}}
\end{subfigure}
\caption{Simulation of the double integrator \eqref{eq: db_int} comparing the \eqref{eq:bcbf-qp} (black), the function-based blended controller in \eqref{eq:safe_ble1} with tuning parameter \eqref{eq:blend_1}-\eqref{eq:tun_par} (pink) and the proposed~\eqref{eq:blend-qp} approach (green dashed). The blended controller experiences undesirable oscillations at the boundary of the safe region (\textbf{top}), while our approach does not. In this scenario, our approach is equivalent to the \eqref{eq:bcbf-qp} (\textbf{bottom right}), but with significantly less computational overhead.}
\label{fig:db_int_oscillations}
\end{figure}

\Cref{fig:db_int_oscillations} compares the \eqref{eq:bcbf-qp} to the blended controller from \eqref{eq:safe_ble1} using the blending function in \eqref{eq:blend_1}-\eqref{eq:tun_par}. Initially, the behavior of the controllers is similar; however, at the origin, the blended controller exhibits undesirable oscillations. This is due to the competing goals of the backup and the primary controller, which the blended controller mixes. The primary controller seeks to enter the right half-plane while the backup controller seeks to apply maximum braking. Thus, near the boundary of the safe set, the blended controller modulates the importance between these two control signals, but cannot do so in an optimal manner, as done by the \eqref{eq:bcbf-qp}. This motivates our central contribution to find a simple and computationally efficient method for input-constrained safe control without such undesirable behavior.
\end{example}
\section{Optimal Controller Interpolation}
\label{sec:theory}

As discussed in \Cref{sec:OG_blending}, the main advantage of function-based blending is that it is a \textit{gradient-free} approach (i.e., computing the sensitivity via \eqref{eq: stm_nominal} is not needed) and it does not require solving a QP online. However, the choice of blending function is subjective and may lead to undesirable behavior (see \Cref{ex:db_int_osc}). Therefore, we propose a middle ground between function-based blending and the \eqref{eq:bcbf-qp}, which seeks to select an optimal blending parameter.

We introduce the following optimization problem, where the parameter $\mu$ is minimized subject to conditions on the forward invariance of the implicit set ${\Cbi \subseteq \Cs}$, which we denote as the \textit{optimally interpolated (OI)} controller:
\begin{align*} 
    ~\mu^{\star}(\x) &= \underset{\mu \in [0,1]}{\text{argmin}} \mkern9mu 
     \frac{1}{2} \mu^2 \tag{OI-QP} \label{eq:blend-qp} \\
    \text{s.t.  }
    &\dot{h}(\phinom, \boldsymbol{k}_{\rm m}(\x)) \geq -\alpha \big( h(\phinom) \big),\ \forall \tau \in [0,T], \\
    &\dot{h}_{\rm b}(\phinomT, \boldsymbol{k}_{\rm m}(\x)) \geq -\alpha_{\rm b} \big( h_{\rm b}(\phinomT) \big).
\end{align*}
Based on~\eqref{eq:safe_ble1}, the total derivatives are computed as
\begin{subequations} \label{eq:hdot_blended}
\begin{align} 
    \dot{h}(&\phinom, \boldsymbol{k}_{\rm m}(\x)) = \dot{h} \big( \phinom, \udes(\x) \big) \\
    & + \mu \Big( \dot{h} \big( \phinom, \ub(\x) \big) - \dot{h} \big( \phinom, \udes(\x) \big) \Big), \nonumber \\
    \dot{h}_{\rm b}(&\phinomT, \boldsymbol{k}_{\rm m}(\x)) = \dot{h}_{\rm b} \big( \phinomT, \udes(\x) \big) \\
    & + \mu \Big( \dot{h}_{\rm b} \big( \phinomT, \ub(\x) \big) - \dot{h}_{\rm b} \big( \phinomT, \udes(\x) \big) \Big), \nonumber
\end{align} 
\end{subequations}
leading to a QP with affine constraints in $\mu$.

Similar to the~\eqref{eq:bcbf-qp}, the~\eqref{eq:blend-qp} can be shown to be feasible and guarantee safety.
\begin{lemma} \label{lemma:mu_qp_feasible}
    There exist class-$\mathcal{K}_{\infty}$ functions $\alpha$, $\alpha_{\rm b}$ such that the \eqref{eq:blend-qp} is feasible for all ${\x \in \Cbi}$.
\end{lemma}
\begin{proof}
    Based on~\eqref{eq:hdot_blended}, the constraints of the~\eqref{eq:blend-qp} reduce to those in~\eqref{eq:backup_constraint} for ${\mu=1}$.
    Therefore, based on Lemma~\ref{lemma: backup_constraint}, there exist $\alpha$, $\alpha_{\rm b}$ such that ${\mu=1}$ satisfies the constraints of the optimization problem, which implies feasibility.
\end{proof}

\begin{theorem} \label{thm:safety_oi}
    Consider the system \eqref{eq:affine-dynamics} with a
    primary controller $\udes$,
    a backup controller $\ub$,
    and the set $\Cbi$ in~\eqref{def:C_BI}.
    The 
    controller $\boldsymbol{k}^{\star}_{\rm m} : \mathcal{X} \rightarrow \mathcal{U}$ given by
    \begin{align}
    \label{eq:optimal_blended_controller}
        \boldsymbol{k}^{\star}_{\rm m}(\x) = (1 - \mu^{\star}(\x) )\udes(\x) + \mu^{\star}(\x) \ub(\x),
    \end{align}
    where ${\mu^{\star}(\x)}$ satisfies the~\eqref{eq:blend-qp},
    renders ${\Cbi \subseteq \Cs}$ forward invariant and satisfies the input constraint, ${\boldsymbol{k}^{\star}_{\rm m}(\x) \in \mathcal{U}}$ for all ${\x \in \mathcal{X}}$.
\end{theorem}
\begin{proof}
    The forward invariance of $\Cbi$ follows from Theorem~\ref{thm: backup_cbf}, since the controller in~\eqref{eq:optimal_blended_controller} satisfies the constraints in~\eqref{eq: backupCBFMain} by construction.
    The controller in~\eqref{eq:optimal_blended_controller} satisfies the input constraint because for all $\x \in \mathcal{X}$, ${\udes(\x) \in \mathcal{U}}$, ${\ub(\x) \in \mathcal{U}}$, ${\mu^\star(\x) \in [0,1]}$, and $\mathcal{U}$ is convex.
\end{proof}

While maintaining feasibility and safety, the~\eqref{eq:blend-qp} offers major computational advantages over the~\eqref{eq:bcbf-qp}.
Namely, the~\eqref{eq:blend-qp} can be solved in closed-form, which eliminates the need for implementing a high-dimensional QP solver.
Moreover, the number of ODEs to be integrated at each time step can be reduced by using  the~\eqref{eq:blend-qp}.
At the same time, the~\eqref{eq:blend-qp} can also eliminate the undesirable behavior observed for the function-based blending approach.
We demonstrate these properties next.
We remark that the price of these benefits is the restriction of the control input to the linear combination of the primary and backup control signals as in~\eqref{eq:optimal_blended_controller} as opposed to the~\eqref{eq:bcbf-qp} that offers a point-wise optimal input from the entire input set $\mathcal{U}$.

\subsection{Closed-form Solution for Optimal Interpolation}
First, we present the solution of the~\eqref{eq:blend-qp} in closed-form.
As a preliminary, we introduce the notation:
\begin{subequations} \label{eq:a_b_coefficients}
\begin{align}
    & a(\tau, \x) \!\triangleq\! \dot{h} \big( \phinom, \udes(\x) \big) \!+\! \alpha \big( h(\phinom) \big), \\
    & b(\tau, \x) \!\triangleq\! \dot{h} \big( \phinom, \ub(\x) \big) \!-\! \dot{h} \big( \phinom, \udes(\x) \big), \\
    & a_{\rm b}(\x) \!\triangleq\! \dot{h}_{\rm b} \big( \phinomT, \udes(\x) \big) \!+\! \alpha_{\rm b} \big( h_{\rm b}(\phinomT) \big), \\
    & b_{\rm b}(\x) \!\triangleq\! \dot{h}_{\rm b} \big( \phinomT,\! \ub(\x) \big) \!-\! \dot{h}_{\rm b} \big( \phinomT,\! \udes(\x) \big)\!,\!\!
\end{align} 
\end{subequations}
which simplifies the constraints to
${a(\tau,\x) + b(\tau,\x) \mu \geq 0}$ and
${a_{\rm b}(\x) + b_{\rm b}(\x) \mu \geq 0}$.
Then, as for the \eqref{eq:bcbf-qp}, we discretize the constraints on the safety of the backup flow
by introducing
${a_i(\x) = a(\tau_i,\x)}$ and
${b_i(\x) = b(\tau_i,\x)}$
with ${\tau_i = i \Delta}$, ${\Delta = T/N}$, ${i \in \{0, 1, \dots, N\}}$.
This leads to
\begin{align} 
    ~\mu^{\star}(\x) = \underset{\mu \in \R}{\text{argmin}} \mkern9mu 
    & \frac{1}{2} \mu^2
    \label{eq:blend-qp-discrete} \\
    \text{s.t.  }
    & a_i(\x) + b_i(\x) \mu \geq 0, \nonumber \
    \forall i \in \mathcal{I}, \nonumber
\end{align}
where
${\mathcal{I} = \{0, 1, \dots, N+3\}}$,
whereas
${a_{N+1}(\x) = a_{\rm b}(\x)}$, 
${b_{N+1}(\x) = b_{\rm b}(\x)}$,
while the domain ${\mu \in [0,1]}$ is captured by
${a_{N+2}(\bx) = 0}$,
${b_{N+2}(\bx) = 1}$,
${a_{N+3}(\bx) = 1}$,
${b_{N+3}(\bx) = -1}$.

\begin{theorem} \label{thm:explicit_solution}
    The QP~\eqref{eq:blend-qp-discrete} has the closed-form solution
    \begin{align} \label{eq:qp_solu_max}
        \mu^{\star}(\x) =
        \max_{i \in \mathcal{I}_+(\x)} \Big\{ -\frac{a_i(\x)}{b_i(\x)} \Big\},
    \end{align}
    where ${\mathcal{I}_+(\x) \triangleq \{i \in \mathcal{I} : b_i(\x) > 0 \}}$.
\end{theorem}

\begin{proof}
    We first notice that the constraints of the discretized QP~\eqref{eq:blend-qp-discrete} are a subset of the constraints of the~\eqref{eq:blend-qp}.
    Therefore, the QP~\eqref{eq:blend-qp-discrete} is feasible by \Cref{lemma:mu_qp_feasible} (considering appropriate $\alpha$ and $\alpha_{\rm b})$.
    Because the primal problem in \eqref{eq:blend-qp-discrete} is a QP and is feasible, Slater's conditions hold and thus strong duality holds. Therefore, any points (${\mu^{\star},\lambda_i}$) that satisfy the Karush--Kuhn--Tucker (KKT) conditions are primal and dual optimal, and have zero duality gap \cite[Ch. 5]{Boyd_Vandenberghe_2004}. That is, the solution of the QP~\eqref{eq:blend-qp-discrete} must satisfy the KKT conditions:
    \begin{subequations} \label{eq:KKT}
    \begin{align}
        a_i + b_i \mu^{\star} \geq 0, \label{eq:primal_feas}\\
        \lambda_i \geq 0, \label{eq:dual_feas} \\
        \lambda_i(a_i + b_i\mu^{\star}) = 0, \label{eq:stationary} \\
        \mu^{\star} - \sum_{i\in\mathcal{I}}\lambda_i b_i = 0, \label{eq:slackness}
    \end{align}
    \end{subequations}
    for all ${i \in \mathcal{I}}$.
    Note that, throughout this proof, we omit the dependencies on $\x$ to simplify our notation.
    
    For ease of notation, let
    ${j \nspace{-2}\triangleq\nspace{-2} {\rm{argmax}}_{i\in\mathcal{I}_{+}}\{-a_i/b_i\}}$. If multiple indices yield the maximum, $j$ can be any of them.
    We prove that the KKT conditions in \eqref{eq:KKT} are satisfied by 
    \begin{subequations} \label{eq:opt_conds}
    \begin{align}
        \mu^{\star} &= -\frac{a_j}{b_j}, \label{eq:opt1} \\
        \lambda_i &= 0, \nspace{6} \forall \nspace{1} i \nspace{1} \in \mathcal{I} \nspace{2} \backslash \nspace{-1} \{ j \}, \label{eq:opt2} \\
        \lambda_j &= -\frac{a_j}{b^2_j}, \label{eq:opt3}
    \end{align}
    \end{subequations}
    where $\mu^{\star}$ matches~\eqref{eq:qp_solu_max}.
    Note that, by definition of $j$,
    \begin{subequations}
    \begin{align}
        -\frac{a_j}{b_j} &\geq -\frac{a_i}{b_i}, \ \forall i \in \mathcal{I}_{+}, \label{eq:a_b_max} \\
        b_j&>0 \ {\rm and} \ a_j \leq 0, \label{eq:a_b_signs}
    \end{align}
    \end{subequations}
    where ${a_j \leq 0}$ holds because ${a_{N+2}=0}$ and ${b_{N+2}=1}$.
    Then, the expressions in~\eqref{eq:opt_conds} satisfy \eqref{eq:slackness} because
    \begin{align*}
    \sum_{i\in\mathcal{I}}\lambda_i b_i\nspace{-2}\! \overset{\eqref{eq:opt2}}{=} \nspace{-2}\! \lambda_j b_j \nspace{-2}\!\overset{\eqref{eq:opt3}}{=}\nspace{-2}\!
    -\frac{a_j}{b_j}
    \nspace{-2}\! \overset{\eqref{eq:opt1}}{=}\nspace{-2}\!
    \mu^{\star}.
    \end{align*}
    The condition in \eqref{eq:stationary} holds, because \eqref{eq:opt2} gives ${\lambda_i = 0}$ for ${i \nspace{1} \in \mathcal{I} \nspace{2} \backslash \nspace{-1} \{ j \}}$, whereas  \eqref{eq:opt1} gives ${a_i + b_i\mu^{\star} = 0}$ for ${i=j}$.
    The condition in \eqref{eq:dual_feas} holds with \eqref{eq:opt2} and \eqref{eq:opt3}-\eqref{eq:a_b_signs}.
    
    Finally, to show that \eqref{eq:primal_feas} holds,
    we leverage the fact that the QP~\eqref{eq:blend-qp-discrete} is feasible (per \Cref{lemma:mu_qp_feasible}), that is, there exists $\mu$ such that ${a_i + b_i \mu \geq 0}$ for all ${i \in \mathcal{I}}$.
    This implies
    ${a_i \geq 0}$ for ${b_i=0}$,
    ${\mu \leq -a_i/b_i}$ for ${b_i<0}$, and
    ${\mu \geq -a_i/b_i}$ for ${b_i>0}$.
    Therefore, the following must hold for $\mu$ to exist:
    \begin{subequations}
    \begin{align}
        a_i \geq 0 \ &{\rm if}\ b_i = 0, \label{eq:feasiblility_bzero} \\
        -\frac{a_i}{b_i} \geq -\frac{a_j}{b_j} \ &{\rm if}\ b_i<0, \label{eq:feasiblility_bneg}
    \end{align}
    \end{subequations}
    where~\eqref{eq:feasiblility_bneg} ensures that the upper bounds on $\mu$ are higher than the lower bounds.
    Then, \eqref{eq:primal_feas} is satisfied because
    \begin{align*} 
        a_i + b_i \mu^{\star} \nspace{-2}\! \overset{\eqref{eq:opt1}}{=} \nspace{-2}\!
        a_i + b_i \Big(\nspace{-5} -\frac{a_j}{b_j} \Big)
        \geq
        0,
    \end{align*}
    where the last inequality holds for ${b_i=0}$ due to~\eqref{eq:feasiblility_bzero}, for ${b_i>0}$ because of~\eqref{eq:a_b_max}, and also for ${b_i<0}$ due to~\eqref{eq:feasiblility_bneg}.
    Thus, the expressions in \eqref{eq:opt_conds} satisfy all the conditions for primal and dual optimality.
\end{proof}

The solution~\eqref{eq:qp_solu_max} of the QP~\eqref{eq:blend-qp-discrete} and the corresponding controller can be written equivalently as
\begin{equation} \label{eq:qp_solu}
\begin{aligned}
    \boldsymbol{k}^{\star}_{\rm m}(\x) & = \udes(\x) + \mu^{\star}(\x) \big( \ub(\x) - \udes(\x) \big), \\  
    \mu^{\star}(\x) & = \max_{i \in \mathcal{I}} \lambda \big( a_i(\x),b_i(\x) \big),
\end{aligned}
\end{equation}
with $\lambda$ in~\eqref{eq:lambda_qp}, where we used ${a_{N+2}(\x) = 0}$, ${b_{N+2}(\x) = 1}$.
Notice the similarity to the solution~\eqref{eq:safetyfilter} of the~\eqref{eq:cbf-qp}.
While the~\eqref{eq:cbf-qp} combines the primary controller $\udes$ with the input ${(\nabla h(\x) \cdot g(\x))^\top}$, that represents the safest control direction based on $h$, the proposed QP~\eqref{eq:blend-qp-discrete} blends $\udes$ with the backup controller $\ub$ to ensure the satisfaction of input constraints in addition to safety.

By refining the discretization and taking the continuous limit (${N \to \infty}$), the solution~\eqref{eq:qp_solu} of the discretized QP~\eqref{eq:blend-qp-discrete} leads to the closed-form solution of the~\eqref{eq:blend-qp}:
\begin{equation}
\begin{aligned} \label{eq:cont_mustar}
    \mu^{\star}(\x) = \max \Big\{ & \max_{\tau \in [0,T]} \lambda \big( a(\tau,\x),b(\tau,\x) \big), \\
    &\;\; \lambda \big( a_{\rm b}(\x), b_{\rm b}(\x) \big), 0 \Big\},
\end{aligned}
\end{equation}
where $a$, $b$, $a_{\rm b}$, $b_{\rm b}$ are defined in~\eqref{eq:a_b_coefficients}, while $\lambda$ is given by~\eqref{eq:lambda_qp}.

\subsection{Sensitivity Computations for Optimal Interpolation}

Next, we establish that the~\eqref{eq:blend-qp} requires fewer ODEs to be integrated than the~\eqref{eq:bcbf-qp}.
In the form~\eqref{eq:hdot_blended}, the derivatives can be computed \textit{exactly}:
\begin{subequations}
\begin{align} 
    \dot{h} \big( \phinom, \udes(\x) \big) &= \nabla h(\phinom) \cdot \boldsymbol{q}_{\rm p}(\tau,\x), \\
    \dot{h}_{\rm b} \big( \phinomT,\udes(\x) \big) &= \nabla h_{\rm b}(\phinomT) \cdot \boldsymbol{q}_{\rm p}(T,\x), \\
    \dot{h} \big( \phinom, \ub(\x) \big) &= \nabla h(\phinom) \cdot \boldsymbol{q}_{\rm b}(\tau,\x), \\
    \dot{h}_{\rm b} \big( \phinomT,\ub(\x) \big) &= \nabla h_{\rm b}(\phinomT) \cdot \boldsymbol{q}_{\rm b}(T,\x),
\end{align} 
\end{subequations}
where
${\boldsymbol{q}_{\rm p}(\tau,\x) \triangleq {\stmnom \big( f(\x) + g(\x) \udes(\x) \big)}}$
and
${\boldsymbol{q}_{\rm b}(\tau,\x) \triangleq {\stmnom \big( f(\x) + g(\x) \ub(\x) \big)}}$
are the push forward of the dynamics, representing the change of the backup flow under $\udes$ and $\ub$, respectively.

Although these terms involve the sensitivity matrix 
${\stmnom}$,
the key intuition is that we do not need to compute the entire sensitivity matrix by solving \eqref{eq: stm_nominal}.
Instead, we only need to obtain directional components of the sensitivity matrix,
by solving the following initial value problems:
\begin{equation} \label{eq:pushforward}
\begin{aligned}
    \frac{\partial }{\partial \tau} \boldsymbol{q}_{\rm p}(\tau,\x) & = \jac(\phinom) \boldsymbol{q}_{\rm p}(\tau,\x),
    \nspace{6} \boldsymbol{q}_{\rm p}(0,\x) = f_{\rm p}(\x), \\
    \frac{\partial }{\partial \tau} \boldsymbol{q}_{\rm b}(\tau,\x) & = \jac(\phinom) \boldsymbol{q}_{\rm b}(\tau,\x),
    \nspace{6} \boldsymbol{q}_{\rm b}(0,\x) = f_{\rm cl}(\x),
\end{aligned}
\end{equation}
where ${f_{\rm p}(\x) \triangleq f(\x) + g(\x) \udes(\x)}$, while $f_{\rm cl}$ is given in~\eqref{eq: f_cl},  and $\jac(\x)=\partial f_{\rm cl}(\x)/\partial \x$; cf.~\eqref{eq: stm_nominal}.

Thus, this method requires integrating forward 3$n$ ODEs, given by~\eqref{eq: nomFlow} and~\eqref{eq:pushforward}, rather than the ${n + n^2}$ ODEs required by the standard bCBF approach, reducing the number of ODEs if ${n \geq 2}$.
This may save a significant amount of computation for high-dimensional systems.

The downside to this approach (and the approach in \Cref{sec:OG_blending}) is that the components of the backup and primary controllers are mixed with equal weighting. To improve performance, one could construct an optimization problem similar to that posed in \eqref{eq:blend-qp}, but for a multi-dimensional blending variable, ${\boldsymbol{\mu} = \begin{bmatrix} \mu_1~ \mu_2~ \ldots~ \mu_m\end{bmatrix}^\top}$ with the corresponding mixed controller ${\boldsymbol{k}_{\rm m}(\x) = \udes(\x) + \Delta_{\boldsymbol{u}}(\x)\boldsymbol{\mu}(\x)}$ where ${\Delta_{\boldsymbol{u}}(\x) \triangleq \diag(\ub(\x) - \udes(\x))}$. Unlike the one-dimensional case, this setup has no simple closed-form solution, but requires standard iterative QP solving methods to solve for the optimal $\boldsymbol{\mu}$ vector (see e.g., active set methods \cite[Ch. 16]{Wright_optimization}).
Thus, we opt for the interpolation approach, as it can be solved in closed-form, providing significant computational savings.
\section{Simulation Results}
\label{sec:sim}

To demonstrate the effectiveness of the proposed approach, we consider two simulation examples: the running double integrator example, and a nonlinear fixed-wing aircraft model\footnote{Code is available at: \href{https://github.com/davidvwijk/OI-CBF}{https://github.com/davidvwijk/OI-CBF}}.

\subsection{Double Integrator}

We consider the same setup as in \Cref{ex:db_int_osc}, and compare the optimally interpolated controller \eqref{eq:optimal_blended_controller} with the~\eqref{eq:bcbf-qp} from \Cref{sec:bCBF} and the blended controller~\eqref{eq:safe_ble1} from \Cref{sec:OG_blending}. \Cref{fig:db_int_oscillations} demonstrates that our approach eliminates the undesired oscillations at the boundary of $\Cs$. While the blended approach subjectively selects $\mu$, our approach chooses the \textit{optimal} value, which yields a safe controller with better performance. We note that for this very special case, the \eqref{eq:bcbf-qp} and the \eqref{eq:blend-qp} are equivalent. Because there is a single input, and the primary and backup controllers are state-independent signals on either extreme of the domain $\mathcal{U}$, the \eqref{eq:blend-qp} can search over the entire domain of inputs to find a feasible, safe controller.

\subsection{Fixed-Wing Aircraft}

Next, consider a nonlinear fixed-wing aircraft model \cite{stephens2021realtime}:
\begin{align} 
    \!\underbrace{\begin{bmatrix}
      \dot{\phi}\\
      \dot{\theta}\\
      \dot{\psi}\\
      \dot{p}_{\rm N} \\
      \dot{p}_{\rm E} \\
      \dot{H} \\
      \dot{P} \\
      \dot{N}_{\rm z}
    \end{bmatrix}}_{\dot{\x}}
    \!=\! 
    \underbrace{\begin{bmatrix}
      P + \frac{N_{\rm z} g_{\rm D}}{V_{\rm T}} \sin \phi \tan \theta,\\
      \frac{g_{\rm D}}{V_{\rm T}}\left(N_{\rm z} \cos \phi - \cos \theta \right)\\
      \frac{N_{\rm z} g_{\rm D} \sin \phi}{V_{\rm T} \cos \theta}\\
      V_{\rm T} \cos \theta \cos \psi \\
      V_{\rm T} \cos \theta \sin \psi\\
      V_{\rm T} \sin \theta \\
      -\frac{1}{\tau_{\rm p}}P \\
      -\frac{1}{\tau_{\rm z}}N_{\rm z}
    \end{bmatrix}}_{f(\boldsymbol{x})}
    \!+\!
    \underbrace{\begin{bmatrix}
        0 & 0 \\ 
        0 & 0 \\
        0 & 0 \\
        0 & 0 \\
        0 & 0 \\
        0 & 0 \\
        \frac{1}{\tau_{\rm p}} & 0 \\
        0 & \frac{1}{\tau_{\rm z}} \\
    \end{bmatrix}}_{g(\x)}
    \underbrace{\begin{bmatrix}
        u_{\rm 1} \\
        u_{\rm 2}
    \end{bmatrix}}_{\boldsymbol{u}},
    \label{eq:aircraft}
\end{align}
which is derived by assuming that the angle of attack and the angle of side slip are zero, the load factor dynamics and roll mode are modeled as a first-order system, and that gravity and load factor are the only accelerations on the aircraft.
In \eqref{eq:aircraft}, the states $\phi$, $\theta$, $\psi$, $p_{\rm N}$, $p_{\rm E}$, $H$, and $P$, and $N_{\rm z}$ represent the roll, pitch, and yaw angles, the north and east position coordinates, altitude, roll rate, and normal load factor, respectively. The constants $V_{\rm T}$, $g_{\rm D}$, $\tau_{\rm p}$ and $\tau_{\rm z}$ are the true airspeed, gravitational acceleration, and the time constants for the roll mode and load factor dynamics, respectively. The bounded inputs are the commanded roll rate ${{u_{1} \in [-\frac{\pi}{2},\frac{\pi}{2}]} \nspace{3}{\rm rad/s}}$ and commanded load factor ${u_{2} \in [-1, 4] \nspace{3}}$ (representing normal accelerations measured in $g_{\rm D}$ that pull the aircraft up).

\begin{figure}[t]
\centering
\begin{minipage}{0.35\textwidth}
    \centering
    \includegraphics[width=\linewidth]{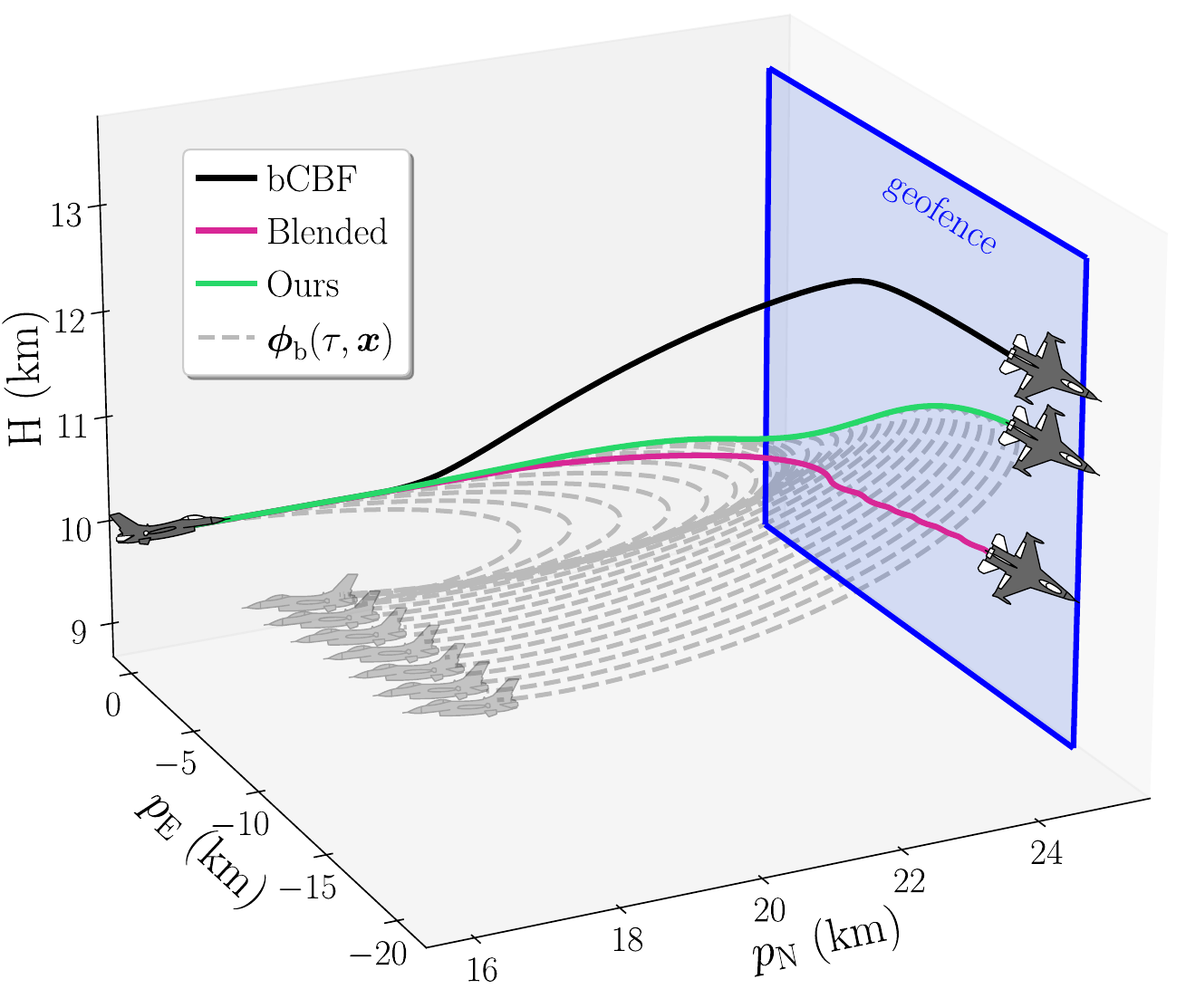}
    \label{subfig:3d}
\end{minipage}%
\hspace{.01cm}
\begin{minipage}{0.12\textwidth}
    \centering
    \includegraphics[width=\linewidth]{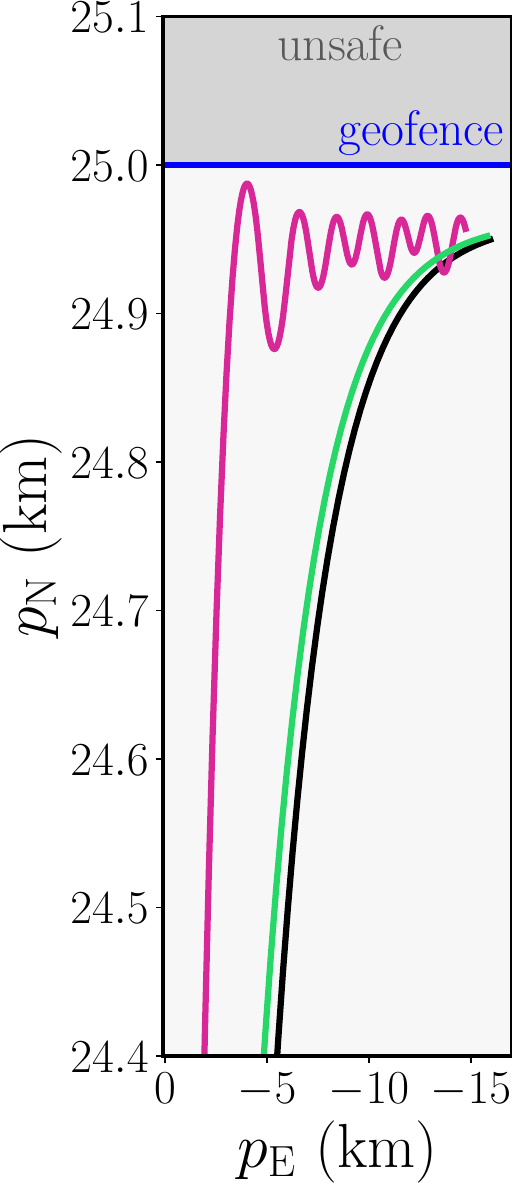}
    \label{subfig:projection}
\end{minipage}%
\hspace{.01cm}
\begin{minipage}{0.48\textwidth}
    \centering
    \vspace{-.15cm}
    \includegraphics[width=\linewidth]{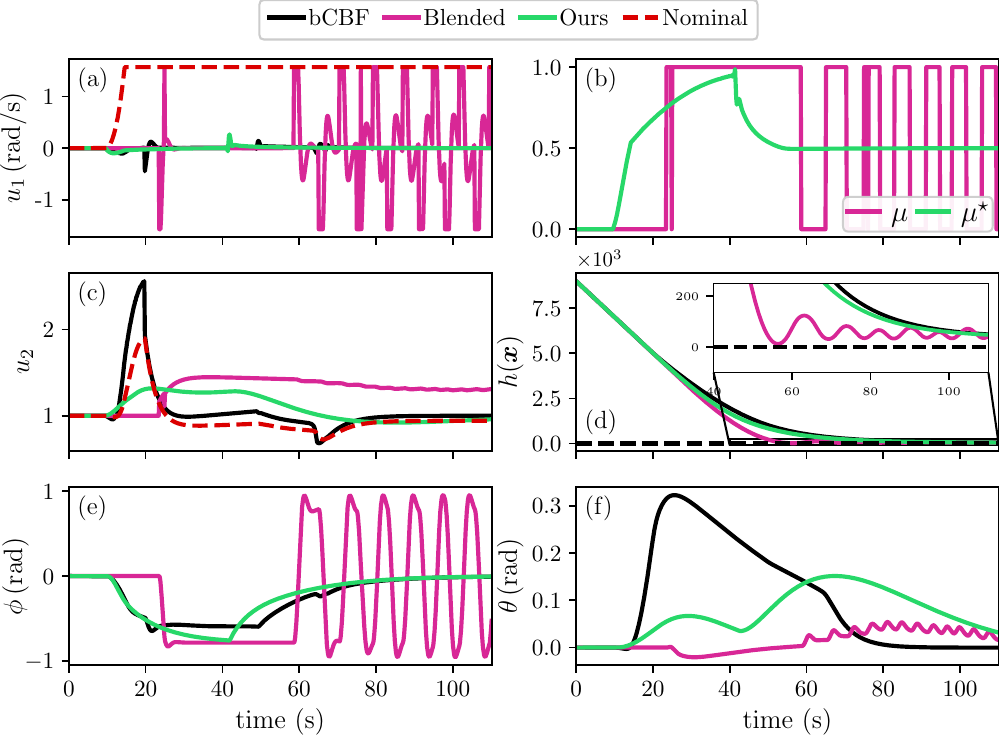}
    \label{subfig:quad}
\end{minipage}
\vspace{-.2cm}
\caption{Simulation of the fixed-wing aircraft~\eqref{eq:aircraft} comparing the \eqref{eq:bcbf-qp} (black), the function-based blended controller in \eqref{eq:safe_ble1} (pink) and the proposed~\eqref{eq:blend-qp} approach (green). The \eqref{eq:bcbf-qp} and \eqref{eq:blend-qp} ensure the safety of \eqref{eq:aircraft} by avoiding crossing the geofence (blue) and gliding near its surface (\textbf{top left}). The blended controller also ensures safety but experiences oscillations along the geofence (\textbf{top right}), causing hysteresis most notably in the roll rate command (\textbf{a}) as the blending parameter switches between $0$ and $1$ (\textbf{b}). This in turn produces stark oscillations in the aircraft's roll angle, which is highly undesirable for stable flight (\textbf{e}). In contrast, the \eqref{eq:blend-qp} smoothly interpolates between $\ub$ and $\boldsymbol{k}_{\rm p}$ (\textbf{b}) to guarantee the safety of the aircraft (\textbf{d}), whilst avoiding evaluating a QP online as required by the \eqref{eq:bcbf-qp}. Further, the closed-form solution from \Cref{thm:explicit_solution} always yields a safe control signal which respects the input bounds, as per \Cref{lemma:mu_qp_feasible}.}
\label{fig:aircraft_sim}
\end{figure}

We consider a scenario where the aircraft must stay within a prespecified airspace, described by a geofence with a known location ${\boldsymbol{p}_{\rm g} \in \mathbb{R}^2}$ and normal vector ${{\boldsymbol{n}_{\rm g}} \in \R^2}$ with ${\norm{\boldsymbol{n}_{\rm g}} = 1}$. The safe set can be written as 
\begin{align} \label{eq:aircraft_Cs}
    \Cs \triangleq \{\boldsymbol{x} \in \mathcal{X} : h(\x) = \boldsymbol{n}_{\rm g} \cdot (\boldsymbol{p} - \boldsymbol{p}_{\rm g}) \geq 0 \},
\end{align}
where ${\boldsymbol{p} = [p_{\rm N}~ \nspace{2} p_{\rm E}]^{\top}}$. The nominal controller seeks to push the aircraft into the geofence while the backup controller executes a coordinated turn with ${\ub(\x) = [u_{\rm b1}~ u_{\rm b2}]^\top}$ where
\begin{align} \label{eq:aircraft_kb}
    u_{\rm b1} &= {\rm sat} \big( P + \tau_{\rm p}\big(K_{\phi}(\phi^* - \phi) - K_{\rm p}P\big) \big),\\
    u_{\rm b2} &= {\rm sat} \big( N_{\rm z} + \tau_{\rm z} \big(K_{\rm N}(N^*_{\rm z} \!-\! N_{\rm z}) + K_{\rm H}(H^* \!-\! H) -K_{\theta}\theta\big) \big), \nonumber
\end{align}
and ${K_{\phi},K_{\rm p},K_{\rm N},K_{\rm H},K_{\theta} > 0}$ are gains, $H^*$ is the altitude, ${\phi^* \in [-\frac{\pi}{4},\frac{\pi}{4}]}$ is the bank angle, and ${N^*_{\rm z} = 1/\cos(\phi^*)}$ is the load factor in the turn. Note that $u_{\rm b1}$ and $u_{\rm b2}$ are saturated using a softplus function which retains the smoothness of $f_{\rm cl}$. Such a controller drives the aircraft to a circular path and renders the following set forward invariant:
\begin{align} \label{eq:aircraft_Cb}
    \Cb \triangleq \left\{ \x \in \mathcal{X} \,\middle|\, 
    \begin{array}{c}
    \phi =\phi^*,\ \theta = 0,\ H =H^*,  \\ P=0,\ N_{\rm z}= N_{\rm z}^*,\ h_6(\x) \geq 0    \\
    \end{array}
    \right\},
\end{align}
with
\begin{equation*}
    h_6(\x) \triangleq h(\x) + \rho\big(\boldsymbol{n}_{\rm g} \cdot \boldsymbol{n}(\psi) - 1\big) - c_6, \nspace{6}
    \boldsymbol{n}(\psi) \triangleq
    \begin{bmatrix}
        -\sin\psi \\ \cos\psi
    \end{bmatrix}.
\end{equation*}
Here,
${\rho = {V_{\rm T}^2}/({g_{\rm D} \tan(\phi^*))}}$ is the turning radius under $\ub$ and ${c_6 > 0}$. This backup set describes the set of periodic orbits whose centers are at least ${\rho + c_6}$ away from the geofence, at a height of $H^*$.
\begin{proposition}
\textit{
The backup set $\Cb$ in~\eqref{eq:aircraft_Cb} is forward invariant for the system~\eqref{eq:aircraft} with the backup controller~\eqref{eq:aircraft_kb}, and it satisfies ${\Cb \subseteq \Cs}$ for the safe set in~\eqref{eq:aircraft_Cs}.
}
\begin{proof}
For all ${\x \in \Cb}$, the states ${\phi, \theta, H, P, \text{ and } N_z}$ are constant along the closed-loop backup dynamics under the backup control law \eqref{eq:aircraft_kb}. The backup flow of the yaw angle evolves with ${\frac{\partial}{\partial \tau}\psi_{\rm b}(\tau,\x) = V_{\rm T}/\rho}$ which has the solution ${\psi_{\rm b}(\tau,\x) = (V_{\rm T}/\rho)\tau + \psi}$. Similarly, the closed-form solution for the backup flow of the position of the aircraft is
\begin{align} \label{eq:position_backup}
    \boldsymbol{p}_{\rm b}(\tau,\x) = \boldsymbol{p} 
    + \rho
    \big( \boldsymbol{n}(\psi) - \boldsymbol{n}(\psi_{\rm b}(\tau,\x)) \big).
\end{align}
Thus, using \eqref{eq:position_backup}, it can be verified that ${h_6(\boldsymbol{p}_{\rm b}(\tau,\x))\geq0}$ along the backup flow for all ${\x \in \Cb}$ and ${\tau \geq 0}$:
\begin{align*}
    h_6 \big( \boldsymbol{p}_{\rm b}(\tau,\x) \big)\!
    &= 
    \boldsymbol{n}_{\rm g} \cdot \big( \boldsymbol{p}_{\rm b}(\tau,\x) - \boldsymbol{p}_{\rm g} \big) \\
    & \quad + \rho \big( \boldsymbol{n}_{\rm g} \cdot \boldsymbol{n}(\psi_{\rm b}(\tau,\x)) - 1 \big) - c_6 \\
    &= \boldsymbol{n}_{\rm g} \cdot (\boldsymbol{p} - \boldsymbol{p}_{\rm g}) + \rho \big( \boldsymbol{n}_{\rm g}  \cdot \boldsymbol{n}(\psi) - 1 \big) - c_6 \\
    &= h_6(\x) \geq 0.
\end{align*}
Further,
because
${\rho\big(\boldsymbol{n}_{\rm g} \cdot \boldsymbol{n}(\psi) - 1\big) \leq 0}$
and ${-c_6 \leq 0}$,
we have ${h(\x) \geq h_6(\x)}$ for all ${\x \in \mathcal{X}}$, and thus ${\Cb \subseteq \Cs}$.
\end{proof}
\end{proposition}
In practice, a neighborhood of $\Cb$ may be used as a backup set. Such a set can be constructed via
${h_1(\x) = c^2_1 \!-\! (\phi^* \!-\! \phi)^2}$,
${h_2(\x) = c^2_2 \!-\! \theta^2}$,
${h_3(\x) = c^2_3 \!-\! (H^* \!-\! H)^2}$,
${h_4(\x) = c^2_4 \!-\! P^2}$,
${h_5(\x) = c^2_5 \!-\! (N^*_{\rm z} \!-\! N_{\rm z})^2}$
with constants
${c_1, c_2, c_3, c_4, c_5 \!>\! 0}$.
Namely, one may use a smooth under-approximation of the minimum of these functions as in \cite{lindemann_stlcbf_2019,tamas_composing23}: 
\begin{align*}
\widehat{\mathcal{C}}_{\rm B} \triangleq \bigg \{ \x \in \mathcal{X} : h_{\rm b}(\x) = {-\frac{1}{\kappa} \ln \Big(\nspace{-2} \sum_{i \in I} {\rm e}^{-\kappa h_i(\x)}\Big)} \geq 0 \bigg \},
\end{align*}
with smoothing parameter ${\kappa > 0}$ and ${I = \{1,\dots,\nspace{-2}6\}}$. 

The nominal controller uses \eqref{eq:aircraft_kb} but with a desired roll angle
${\phi^* = {\rm sat}( K_{\psi} (\psi^* - \psi) V_{\rm T}/{g_{\rm D}})}$
and yaw angle
${\psi^* = \arctan((p^*_{\rm E} - p_{\rm E})/(p^*_{\rm N} - p_{\rm N})) \in \mathbb{S}^1}$,
where ${K_{\psi} > 0}$ and $p^*_{\rm E}$, $p^*_{\rm N}$ is a desired setpoint which is purposefully placed far behind the geofence.

\Cref{fig:aircraft_sim} shows the simulation results for system~\eqref{eq:aircraft} comparing the methods from \Cref{sec:bCBF,sec:OG_blending} with the proposed~\eqref{eq:blend-qp} controller\footnote{The simulation uses
${{g_{\rm D}} \!=\! 9.81 \nspace{2}{\rm m/s^2}}$,
${{V_{\rm T} \!=\! 200 \nspace{2}{\rm m/s}}}$,
${{\tau_{\rm p} \!=\! 1 \nspace{2}{\rm s}}}$,
${{\tau_{\rm z} \!=\! 1 \nspace{2}{\rm s}}}$,
${{\phi^* \!=\! - \frac{\pi}{4} \nspace{2}{\rm rad}}}$,
${H^* \!=\! 10 \nspace{2}{\rm km}}$.
}. The \eqref{eq:blend-qp} guarantees the safety of the aircraft by smoothly selecting the parameter $\mu^{\star}(\x)$ which minimizes the utilization of $\ub$, whilst ensuring that \eqref{eq:aircraft} stays within the controlled invariant set $\Cbi$ defined in \eqref{def:C_BI}. Thanks to \Cref{thm:safety_oi} and \Cref{thm:explicit_solution}, this controller satisfies the safety constraint, and the optimal blending parameter $\mu^{\star}(\x)$ is obtained in closed-form. Further, this parameter ensures that the safe controller always satisfies input constraints, as can be seen in \Cref{fig:aircraft_sim}.
While the \eqref{eq:bcbf-qp} ensures the safety of \eqref{eq:aircraft} under input bounds, it requires solving a QP with a very large number of constraints online, and forward integrating a large number of ODEs (i.e., ${n^2+n=72}$ ODEs, while the optimal interpolation uses ${3n = 24}$ ODEs). The blended approach \eqref{eq:safe_ble1}-\eqref{eq:tun_par} also obeys the geofence constraint and the input bounds, but experiences hysteresis and subsequent oscillations near the geofence.
\section{Conclusion}
\label{sec:conc}
We presented a novel framework to guarantee the safety of nonlinear control-affine systems by optimally interpolating between a verified backup controller and a nominal controller. 
We derived a closed-form solution for the resulting safe controller, and proved that this controller respects input bounds while guaranteeing safety. 
Through simulations, including a double integrator and a geofencing scenario on a fixed-wing aircraft, we demonstrated the overall performance improvement of our approach over the existing function-based controller blending method. This advancement simplifies the deployment of bCBF-based safety filters, providing feasibility guarantees in computationally limited systems. 
Future work will explore OI-QP implementations on computationally-constrained hardware platforms such as a fixed-wing aircraft.

\bibliographystyle{ieeetr}
\bibliography{refs}

\end{document}